\documentclass[11pt]{article}       

\usepackage{verbatim}
 \usepackage{amsmath}
 \usepackage{mathtools}
\usepackage{graphicx}
\usepackage{amsfonts}
\usepackage{authblk}
\usepackage[T1]{fontenc}

\usepackage{amsthm}

\newtheorem{theorem}{Theorem}[section] 
 
\newcommand{\figwidth}{125mm}
\newcommand{\figwidths}{100mm}

\begin{document}



\title{Estimating phylogenetic distances between genomic sequences based on 
the length distribution of $k$-mismatch common substrings} 

\author{Burkhard Morgenstern, Svenja Sch{\"o}bel, Chris-Andr{\'e} Leimeister}  

\affil{ University of G\"ottingen, Department of Bioinformatics, 
Goldschmidtstr. 1, 37077 G\"ottingen, Germany}  


\maketitle 

\begin{abstract}
Various approaches to alignment-free sequence comparison are based on 
the length of exact or inexact word matches between two input sequences. 
Haubold {\em et al.} (2009) showed  how the average number of substitutions
between two DNA sequences can be estimated based on the average length of 
exact common substrings. In this paper, we study the
length distribution of $k$-mismatch common substrings between two sequences. 
We show that the number of substitutions per position that have occurred 
since two sequences have evolved from their last common ancestor,
can be estimated from the position of a local maximum in the length 
distribution of their $k$-mismatch common substrings.   
\end{abstract}

\section{Introduction}
Phylogenetic distances between DNA or protein sequences are usually
estimated based on pairwise or multiple sequence alignments.
Since sequence alignment is computationally expensive, alignment-free
phylogeny approaches have become popular in recent 
years, see Vinga \cite{vin:14}  for a review. 
Some of these approaches compare the word composition  
\cite{hoe:rig:rag:06,sim:jun:wu:kim:09,cho:hor:lev:09,vin:car:fra:etal:12}
or spaced-word composition
\cite{lei:bod:hor:lin:mor:14,mor:zhu:hor:lei:15,hah:lei:oun:etal:16,noe:17}
 of sequences using a fixed word length or pattern of match and don't-care
positions, respectively.  
Other approaches are based on the {\em matching statistics} \cite{cha:law:94}, 
that is on the length of common substrings of the input
sequences \cite{uli:bur:tul:cho:06,com:ver:12}.  
All these methods are much faster than traditional alignment-based approaches.
A disadvantage of most word-based approaches to phylogeny reconstruction
is that they are not based on explicit models of molecular evolution.
Instead of estimating  distances in a statistically rigorous sense,
they only return rough measures of sequence similarity or dissimilarity.

The {\em average common substring (ACS)} approach  \cite{uli:bur:tul:cho:06} 
calculates for each position in one sequence the length of the longest
substring starting at this position that matches a substring of the
other sequence. The average length of these substring matches is then used
to quantify the similarity between two sequences based on 
information-theoretical 
considerations; these similarity values are finally transformed into symmetric
distance values. 
More recently, we generalized the {\em ACS} approach by considering
common substrings with up to $k$ mismatches instead of exact substring
matches \cite{lei:mor:14}. 
To calculate distance
values between two sequences from the average length of $k$-mismatch common 
substrings, 
we used the same information-theoretical approach as in {\em ACS}.  
Since there is no exact solution to the  
{\em $k$-mismatch longest common substring problem} that is fast enough
to be applied to long genomic sequences, we proposed a simple heuristic: 
we first search for longest common {\em exact} matches and then extend
these matches until the $k+1$st mismatch occurs. Distances are then 
calculated from the {\em average} length of these $k$-mismatch common substrings
similarly as in {\em ACS}; 
the implementation of this approach is called {\em kmacs}. 

Various algorithms have been proposed in recent years 
to calculate exact or approximate
solutions for the {\em $k$-mismatch average common substring problem} 
as a basis for phylogeny reconstruction 
\cite{alu:apo:tha:15,tha:cho:liu:apo:alu:15,piz:16,tha:apo:alu:16,apo:gue:lan:piz:16,piz:16,tha:cho:liu:etal:17,pet:gue:piz:17}.
Like {\em ACS} and {\em kmacs}, these approaches do not estimate 
the `real' pairwise distances between 
sequences in terms of substitutions per position. Instead, they calculate  
various sorts of distance measures that vaguely reflect evolutionary 
distances.  

To our knowledge, the first alignment-free approach to estimate the
phylogenetic distance 
between two DNA sequences in a statistically rigorous way was the
program {\em kr} by Haubold {\em et al.} \cite{hau:pfa:dom:wie:09}. 
These authors showed that the average number of
nucleotide substitutions per position between two DNA sequences  
can be estimated by calculating for each position $i$ in the first sequence
the length of the shortest substring starting at $i$ that does not occur 
in the second sequence, see also \cite{hau:pie:moe:wie:05,hau:wie:06}. 
This way, phylogenetic distances between DNA
sequences can be accurately estimated for distances up to around 
$0.5$ substitutions per position. 
Some other, more recent alignment-free approaches  also estimate 
phylogenetic distances based on a stochastic model of molecular evolution, 
namely {\em Co-phylog} \cite{yi:jin:13},  {\em andi} \cite{hau:klo:pfa:14},
an approach based on the {\em number} of (spaced-) word matches
\cite{mor:zhu:hor:lei:15} and {\em Filtered Spaced Word Matches}
\cite{lei:den:mor:17}.

In this paper, we propose a new  
approach to estimate phylogenetic distances based on the length distribution
of $k$-mismatch common substrings. 
The manuscript is organized as follows. In section~\ref{notation}, we
introduce some notation and the stochastic model of sequence evolution
that we are using. In section~\ref{acs}, we
recapitulate a result from \cite{hau:pfa:dom:wie:09} on the length
distribution of longest common substrings, which we generalize 
in section~\ref{k_mismatch_acs} to $k$-mismatch longest common substrings, and 
in section~\ref{kmacs}, we study the length distribution of $k$-mismatch
common substrings returned by the {\em kmacs} heuristic~\cite{lei:mor:14}.
In sections \ref{distance} and~\ref{implementation}, 
we introduce our new approach to estimate phylogenetic
distances and explain some implementation
details. Finally, sections~\ref{results} and \ref{discussion} 
report on benchmarking results, 
discusses these results and address some possible future developments. 

We should mention that sections \ref{k_mismatch_acs} and \ref{kmacs}
are not necessary to understand our novel approach to distance estimation, 
except for equation (\ref{binom_eq}) which gives the length
distribution of $k$-mismatch common substrings at given positions $i$ and $j$.
We added these two sections for completeness, and since the results could be
the basis for alternative ways to estimate phylogenetic distances. But readers
who are mainly interested in our approach to distance estimation can
skip sections \ref{k_mismatch_acs} and \ref{kmacs}.

\section{Sequence model and notation} 
\label{notation}
We use standard notation such as used in \cite{gus:97}.  
For  a sequence $S$ of length $L$ over some alphabet, $S(i)$ 
is the $i$-th character in $S$. $S[i..j]$ denotes the (contiguous)  
substring from $i$ to $j$; we say  that $S[i..j]$ is a 
{\em substring at $i$}. 
In the following, 
we consider two DNA sequences $S_1$ and $S_2$ that are assumed to have
descended from an unknown common ancestor under the 
{\em Jukes-Cantor} model \cite{juk:can:69}.  That is, we assume that 
substitutions at different positions are independent of each other, that
we have a constant substitution rate at all positions and  that all 
substitutions occur with the same probability. 
%
%
Thus, we have $p$ and $q$ with 
\[
   P\left(S_1(i) = S_2(j)\right)  
 = \left\{ 
\begin{array}{ll}
p & \text{ if } i = j\\ 
q & \text{ else } \\ 
\end{array}
\right.
\] 
Moreover, we use a gap-free model of evolution to simplify the 
considerations below.  
Note that, with a gap-free model, it is trivial to estimate the number
of substitutions since two sequences diverged from their last common
ancestor,  
 simply by counting the number of mismatches in the gap-free alignment and
then applying the usual Jukes-Cantor correction. 
However, we will to apply this simple model to real-world sequences with 
insertions and deletions where this trivial approach is not possible.  

\section{Average common substring length}
\label{acs} 
For positions $i$ and $j$ in sequence $S_1$ and  $S_2$, respectively, 
we define random variables 
\[X_{i,j} = \max\{l: X[i..i+l-1] = X[j..j+l-1] \} \]  
as the length of the longest substring 
at $i$ that exactly matches a substring at $j$.  Next, we define 
\[X_i = \max_{1\le j\le L} X_{i,j}\]
as the length of the {\em longest substring} at $i$ that matches a 
substring of $S_2$. 

In the following, we ignore edge effects which is justified if 
long sequences are compared since the probability of $k$-mismatch common 
substrings of length $m$ decreases rapidly if $m$ increases. 
With this simplification, we  have
\[ 
 P(X_{i,j} <n )  = 
 1 - P(X_{i,j} \ge n )  = 
\left\{ 
\begin{array}{ll}
1-p^n & \text{ if } i = j \\
1-q^n & \text{ else }  \\
\end{array}
\right. 
\]
If, in addition, 
we assume equilibrium frequencies for the nucleotides, {\em i.e.} if we
assume that each nucleotide occurs at each sequence position with probability
$0.25$,  
the random variables $X_{i,j}$ and $X_{i',j'}$ are independent of each
other whenever $j-i\not=j'-i'$ holds. 
In this case, we have for $n\le L-i+1$ 

\begin{equation}
\label{X_i_l_n}
\begin{aligned}
\MoveEqLeft
{P(X_i < n) = P(X_{i,1} < n \wedge \ldots \wedge X_{i,L} < n )}\\  
        &= P(X_{i,1} < n) \cdot \ldots \cdot P(X_{i,L} < n) \\ 
	& = (1-q^n)^{L-1} \cdot (1-p^n) 
\end{aligned}
\end{equation}
and 
\begin{equation} 
\nonumber 
\begin{aligned}
\MoveEqLeft{P(X_i = n) = P(X_i < n+1) - P(X_i < n) }\\
	& = (1-q^{n+1})^{L-1} \cdot (1-p^{n+1})   
 - (1-q^n)^{L-1} \cdot (1-p^n)  
\end{aligned}
\end{equation}
so the expected length of the longest common substring at 
a given sequence position is  
\begin{eqnarray} 
\label{1st_E_approx}
E(X) = 
\sum_{n=1}^L
n \cdot 
\left(
(1-q^{n+1})^{L-1} \cdot (1-p^{n+1})  
 - (1-q^n)^{L-1} \cdot (1-p^n) 
\right)
\end{eqnarray}  


\section{$k$-mismatch average common substring length}
 \label{k_mismatch_acs} 
Next, we generalize the above considerations by considering 
the average length of the 
$k$-{\em mismatch longest common substrings} between two sequences 
for some integer $k \ge 0$. That is, for a position $i$ in one
of the sequences, we consider the longest substring starting at $i$ that
matches some substring in the other sequence with a {\em Hamming distance} 
$\le k$.  Generalizing the above notation, we define random variables 
\[X_{i,j}^{(k)} = 
\max \left\{l: d_H\left(S_1[i..i+l-1],S_2[j..j+l-1]\right) \le k\right\} \] 
where $d_H(\cdot,\cdot)$ is 
the {\em Hamming distance} between two sequences. 
In other words, $X_{i,j}^{(k)}$
is the length 
of the longest substring starting at position $i$ in sequence $S_1$
that matches a substring starting at position $j$ in sequence $S_2$ with
to $k$ mismatches. Accordingly, we define 
\[ X_{i}^{(k)} = \max_j X^{(k)}_{i,j}\] 
as the length of the longest $k$-mismatch substring at position $i$. 
As pointed out by Apostolico {\em et al.}~\cite{apo:gue:lan:piz:16}, 
 $X^{(k)}_{i,j}$ follows a {\em negative binomial distribution}.  
More precisely, we have
$X^{(k)}_{i,j} \sim NB(n-k;k-1,p)$, and we can write 
\begin{equation}
\label{binom_eq}
	P\left( X^{(k)}_{i,j}  =  n \right) =  
	\left\{  
	  \begin{array}{ll} 
		  {n \choose k} p^{n-k} (1-p)^{k+1}  & \text{ if } i=j \\
		  {n \choose k} q^{n-k} (1-q)^{k+1}  & \text{ else }  \\
	  \end{array}
 \right.
\end{equation}
and
\begin{equation}
\label{binom_ge}
P\left( X^{(k)}_{i,j} \ge n \right)=
        \left\{  
          \begin{array}{ll} 
 \sum_{k'\le k} {n \choose k'} p^{n-k'} (1-p)^{k'} & \text{ if } i=j \\ 
 \sum_{k'\le k} {n \choose k'} q^{n-k'} (1-q)^{k'} & \text{ else } \\ 
          \end{array}
 \right.
\end{equation}
Generalizing (\ref{X_i_l_n}), we obtain for $n>k$  
\begin{eqnarray}
	\lefteqn{P\left(X_{i}^{(k)} < n\right) = }\\  
	  \nonumber 
	&& \left(1 - \sum_{k'\le k} {n \choose k'} q^{n-k'} 
 (1-q)^{k'}\right)^{L+i-1} \cdot 
\left( 1 - \sum_{k'\le k} {n \choose k'} p^{n-k'} (1-p)^{k'}\right)
\end{eqnarray}
while we have 
$$ P\left(X_{i}^{(k)} < n\right) = 
\left\{
\begin{array}{ll} 
1 & \text{ if } n > L-i+1\\  
0 & \text{ if } n \le k  \\ 
\end{array}
\right. 
$$ 
Finally, we obtain 
\begin{eqnarray}
        \label{eq_expect}
          \nonumber 
        \lefteqn{P\left(X_i^{(k)}=n\right)  
         = \left. \left(1 - \sum_{k'\le k} {n+1 \choose k'} q^{n+1-k'} (1-q)^{k'}\right)^{L+i-1} \right. }\\  
         && \cdot 
        \left( 1 - \sum_{k'\le k} {n+1\choose k'} p^{n+1-k'} (1-p)^{k'}\right)
        \\ 
        &-& 
        \nonumber 
        \left.
        \left(1 - \sum_{k'\le k} {n \choose k'} q^{n-k'} (1-q)^{k'}\right)^{L+i-1} 
         \cdot 
        \left( 1 - \sum_{k'\le k} {n \choose k'} p^{n-k'} (1-p)^{k'}\right)
        \right. 
\end{eqnarray}
from which one can obtain the expected length of the $k$-mismatch longest 
substrings. 
%
%

%

\section{Heuristic used in {\em kmacs}} 
\label{kmacs}
Since exact solutions for the {\em average $k$-mismatch common substring 
problem} are too time-consuming for  large sequence sets, 
the program {\em kmacs} \cite{lei:mor:14} uses a heuristic.  
In a first step, the program calculates for each position $i$ in one sequence, the length of 
the longest substring starting at $i$ that {\em exactly} matches a substring
of the other sequence. {\em kmacs} then calculates the length of the longest 
gap-free {\em extension} of this exact match with up to $k$ mismatches.  
Using standard indexing structures, this can be done in
$O(L\cdot k)$ time.

For  sequences $S_1, S_2$ as above and a position $i$ in $S_1$, 
  let $j^*$ be a position in $S_2$ such that the 
$X_i$-length substring starting at $i$ matches the $X_i$-length 
substring at $j^*$ in $S_2$. That is, the substring 
\[S_2[j^*..j^* + X_i -1]\]
is the longest substring of $S_2$ that matches a substring of $S_1$ 
at position $i$.   
In case there are several such positions in $S_2$, we assume for simplicity
that $j^* \not= i$ holds 
(in the following, we only need to distinguish the cases
$j^*=i$ and $j^*\not= i$, otherwise it does not matter how $j^*$
is chosen). 
Now, let the random variable 
$\tilde{X}^{(k)}_i$ be defined as the length
of the $k$-mismatch common substring starting at $i$ and $j^*$, so we have

\begin{equation}
\label{sum_kmacs}
\tilde{X}^{(k)}_i = X_{i,j^*}^{(k)}
 = X_i + X^{(k-1)}_{i+X_i,j^*+X_i} + 1   
\end{equation}
%

%
%
%

\begin{theorem}
\label{kmacs_heur_prob_theo}
For a pair of sequences as above, $1 \le i \le L$  and $m\le L + i$,
the probability of the heuristic {\em kmacs} hit of having a length of $m$
is given as
\begin{equation}
\begin{aligned}
\nonumber
\MoveEqLeft {P\left(\tilde{X}^{(k)}_i = m\right)} 
\\ 
=
&\ \  
p^{m-k+1}(1-p)^{k+1}  
\sum_{m_1+m_2=m} 
  (1 - q^{m_1+1})^{L-1}
 {m_2 \choose k-1} \\
&+ \sum_{m_1+m_2=m} \left[(1-q^{m_1+1})^{L-1} - (1-q^{m_1})^{L-1} \right] 
    \cdot (1-p^{m_1}) \\ 
& \ \ {m_2 \choose k-1}  
q^{m_2-k+1}(1-q)^k \\ 
\end{aligned}
\end{equation}
\end{theorem}

\begin{proof}

Distinguishing between `homologous' and `background' matches, we can 
write
\begin{equation}
\label{kmacs_heur_len}
\begin{aligned}
\MoveEqLeft {P\left(\tilde{X}^{(k)}_i = m\right)} 
 =  P\left(\tilde{X}^{(k)}_i = m\middle| j^*=i\right) P(j^*=i) \\ 
   +& P\left(\tilde{X}^{(k)}_i = m\middle| j^*\not=i\right) P(j^*\not=i)  
\end{aligned}
\end{equation}
and with (\ref{binom_eq}), we obtain
\begin{equation}
\label{hom_eq1}
\begin{aligned}
\MoveEqLeft P\left(\tilde{X}^{(k)}_i = m\middle| j^*=i\right) \\
&= \sum_{m_1+m_2=m} P(X_i = m_1 |  j^*=i) P\left(X_{i+m_1,i+m_1}^{(k-1)}=m_2\right) \\  
&= \sum_{m_1+m_2=m} P(X_i = m_1 |  j^*=i) {m_2 \choose k-1}  
p^{m_2-k+1}(1-p)^k 
\end{aligned}
\end{equation}
and 
\begin{equation}
\label{hom_eq2}
\begin{aligned}
\MoveEqLeft P(X_i = m_1 |  j^*=i)  
= \frac{P(X_{i,i}=m_1 \wedge j^*=i)}{P(j^*=i) } \\  
&= \frac{P(X_{i,i}=m_1 \wedge X_{i,i} \ge X_{i,j}, j\not= i)}{P(j^*= i)} \\ 
&= \frac{P(X_{i,i}=m_1 \wedge  X_{i,j}\le m_1, j\not= i)}{P(j^*= i)} \\ 
&= \frac{p^{m_1}(1-p) \cdot (1 - q^{m_1+1})^{L-1}}{P(j^*= i)} \\ 
\end{aligned}
\end{equation}
%
%
%
so with (\ref{hom_eq1}) and (\ref{hom_eq2}), 
the first summand in (\ref{kmacs_heur_len}) becomes
\begin{equation}
\begin{aligned}
\MoveEqLeft  P\left(\tilde{X}^{(k)}_i = m\middle| j^*=i\right) P(j^*=i) \\ 
=& \sum_{m_1+m_2=m} P(X_i = m_1 |  j^*=i) {m_2 \choose k-1}  
p^{m_2-k+1}(1-p)^k \cdot P(j^*=i) \\ 
=& \sum_{m_1+m_2=m} 
\frac{p^{m_1}(1-p) \cdot (1 - q^{m_1+1})^{L-1}}{P(j^*= i)} \\ 
& \ \  {m_2 \choose k-1}  
p^{m_2-k+1}(1-p)^k \cdot P(j^*=i) \\ 
=& \sum_{m_1+m_2=m} 
  (1 - q^{m_1+1})^{L-1}
 {m_2 \choose k-1}  
p^{m_1+m_2-k+1}(1-p)^{k+1} \\ 
=&\ \  
p^{m-k+1}(1-p)^{k+1}  
\sum_{m_1+m_2=m} 
  (1 - q^{m_1+1})^{L-1}
 {m_2 \choose k-1} \\ 
\end{aligned}
\end{equation}

Similarly, for the second summand in (\ref{kmacs_heur_len}), we note that
\begin{equation}
\label{bg_eq1}
\begin{aligned}
\MoveEqLeft P\left(\tilde{X}^{(k)}_i = m\middle| j^*\not=i\right) \\
&= \sum_{m_1+m_2=m} P(X_i = m_1 |  j^*\not=i) {m_2 \choose k-1}  
q^{m_2-k+1}(1-q)^k 
\end{aligned}
\end{equation}
and
\begin{equation}
\label{bg_eq2}
\begin{aligned}
\MoveEqLeft P(X_i = m_1 |  j^*\not=i)  
= \frac{P(X_{i,j^*}=m_1 \wedge j^*\not=i)}{P(j^*\not=i) } \\  
&= \frac{P(X_{i,j^*}=m_1 \wedge X_{i,i} < X_{i,j^*})}{P(j^*\not= i)} \\ 
&= \frac{P(X_{i,j^*}=m_1 \wedge X_{i,i} < m_1 )}{P(j^*\not= i)} \\ 
&= \frac{P(\max_{j\not=i}X_{i,j}=m_1 \wedge X_{i,i} < m_1 )}{P(j^*\not= i)} \\ 
&= \frac{P(\max_{j\not=i}X_{i,j}=m_1) \cdot P( X_{i,i} < m_1 )}{P(j^*\not= i)}\\ 
&= \frac{P(\max_{j\not=i}X_{i,j}=m_1)\cdot P( X_{i,i} < m_1)}{P(j^*\not= i)}\\ 
&= \frac{\left[(1-q^{m_1+1})^{L-1} - (1-q^{m_1})^{L-1} \right] 
    \cdot (1-p^{m_1})}{P(j^*\not= i)}\\ 
\end{aligned}
\end{equation}
Thus, the second summand in (\ref{kmacs_heur_len}) is given as
\begin{equation}
\nonumber
\begin{aligned}
\MoveEqLeft {P\left(\tilde{X}^{(k)}_i = m\middle| j^*\not=i\right) P(j^*\not=i)}\\ 
=& \sum_{m_1+m_2=m} P(X_i = m_1 |  j^*\not=i) {m_2 \choose k-1}  
q^{m_2-k+1}(1-q)^k \cdot P(j^*\not= i) \\  
=& \sum_{m_1+m_2=m} \frac{\left[(1-q^{m_1+1})^{L-1} - (1-q^{m_1})^{L-1} \right] 
    \cdot (1-p^{m_1})}{P(j^*\not= i)} \\ 
& \ \ {m_2 \choose k-1}  
q^{m_2-k+1}(1-q)^k \cdot P(j^*\not= i) \\  
=& \sum_{m_1+m_2=m} \left[(1-q^{m_1+1})^{L-1} - (1-q^{m_1})^{L-1} \right] 
    \cdot (1-p^{m_1}) \\ 
& \ \ {m_2 \choose k-1}  
q^{m_2-k+1}(1-q)^k \\  
\end{aligned}
\end{equation}

\end{proof}

\begin{figure}
\begin{center}
\includegraphics[width=\figwidths]{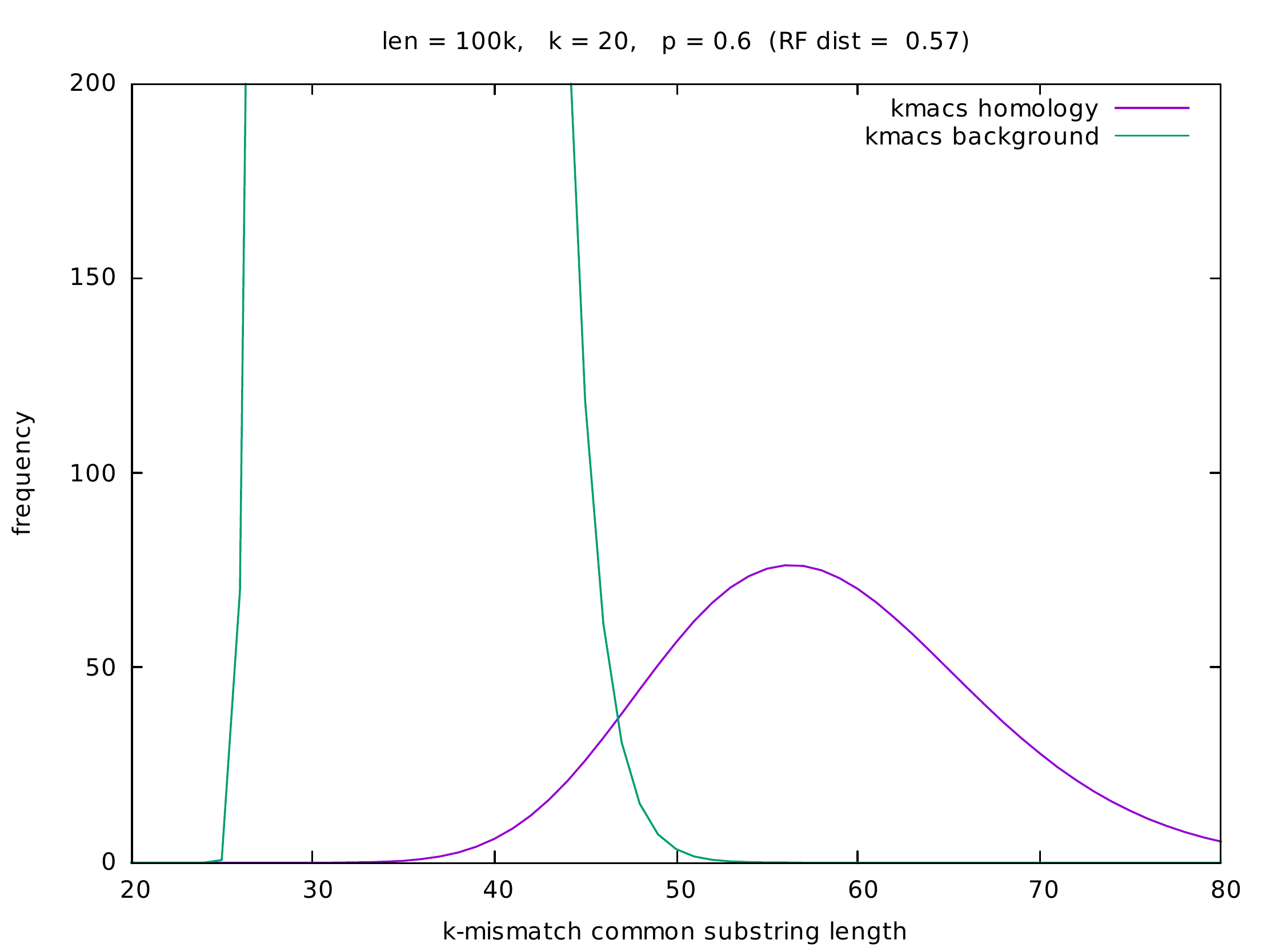}

\includegraphics[width=\figwidths]{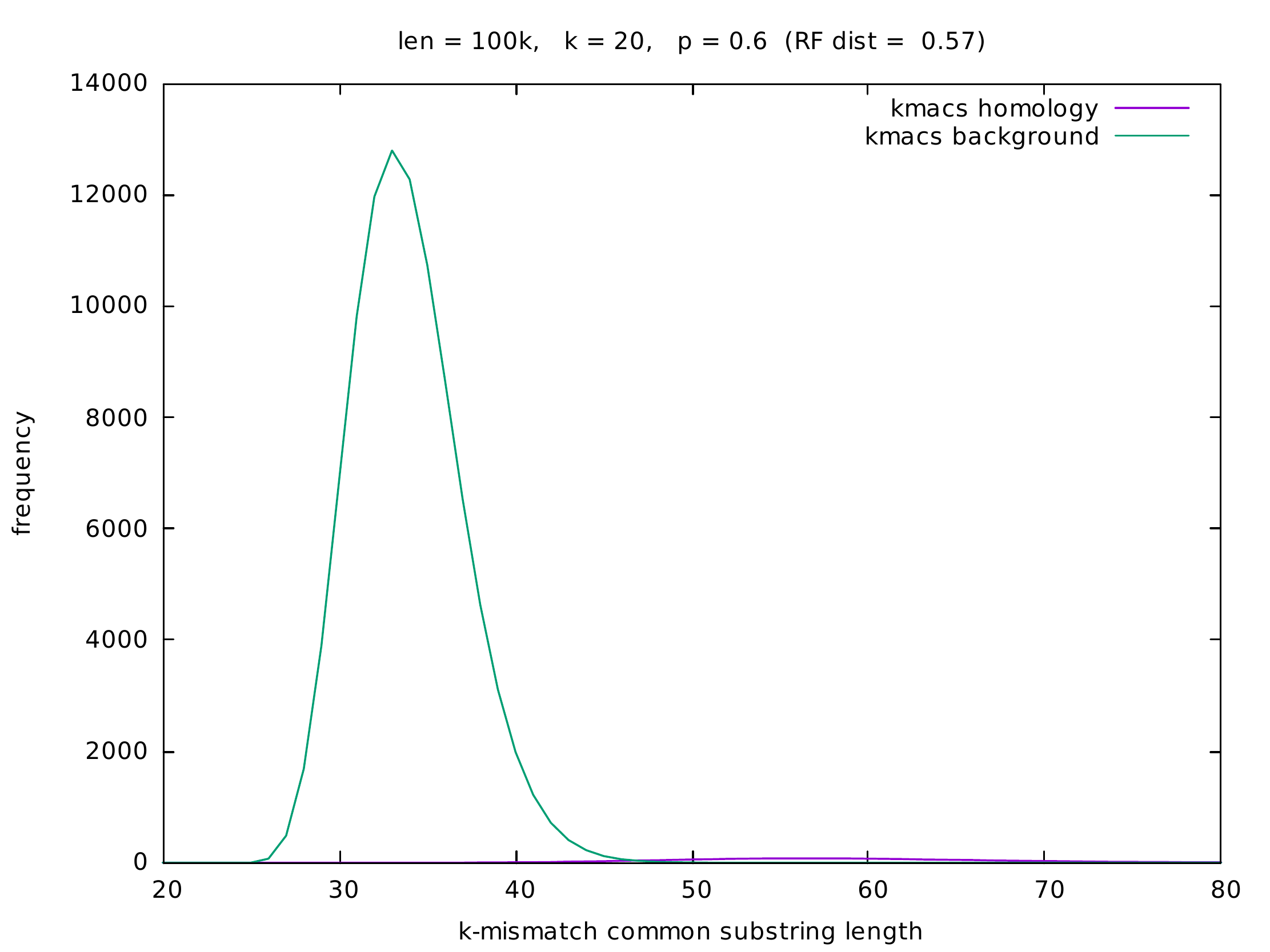}
\end{center}
\caption{\label{overlap_kmacs}Length distribution of the background
and homologous $k$-mismatch longest common substrings for a pair of  
DNA sequences under the {\em Jukes-Cantor} model. 
For each possible length, the {\em expected} 
number of $k$-mismatch longest common 
substrings of this length returned by the {\em kmacs} heuristic 
is calculated using theorem~\ref{kmacs_heur_prob_theo}
 for an indel-free 
pair of sequences of length $L=100kb$, a match probability $p=0.6$ 
(corresponding to 0.57 substitutions per position) and $k=20$. 
}
\end{figure}

For $1\le m \le L$, the expected number of $k$-mismatch common substrings
of length $m$ returned by the {\em kmacs} heuristics is given as  
$L \cdot P\left(\tilde{X}^{(k)}_i = m\right)$ and can be calculated 
using theorem~\ref{kmacs_heur_prob_theo}.
In  Figure~\ref{overlap_kmacs}, these values are plotted against $m$ 
for $L=100$~{\em kb}, $p=0.6$ and $k=20$.

%
%
%
%
%
%

\section{Distance estimation} 
\label{distance}
Using theorem \ref{kmacs_heur_prob_theo}, one could 
estimate the match probability $p$ -- and thereby the average number of
substitutions per position -- from
the {\em empirical} average length of the $k$-mismatch common substrings 
returned by {\em kmacs}  
in a moment-based approach, similar to the approach proposed 
in~\cite{hau:pfa:dom:wie:09}.

A problem with this moment-based approach is that, for realistic values
of $L$ and $p$, one has
$P(j^*=i)  \ll  P(j^*\not=i)$,
so the above sum is heavily dominated by the `background' part, {\em i.e.}
by the second summand in (\ref{kmacs_heur_len}). 
For the parameter values used in Figure~\ref{overlap_kmacs}, for example,
only 1 percent of the matches returned by {\em kmacs} represent homologies
while 99 percent are background noise.  
There are, in principle, two ways to circumvent this problem. First, one could
try to separate homologous from background matches using a suitable 
threshold values, similarly as we have done it in our {\em Filtered 
Spaced Word Matches} approach \cite{lei:soh:mor:17}. But this is more
difficult for $k$-mismatch common substrings, since there is much more
overlap between homologous and background matches than for {\em Spaced-Word}
matches, see Figure~\ref{overlap_kmacs}.

There is an alternative to this moment-based approach, however.
As can be seen in  Figure~\ref{overlap_kmacs}, the length distribution
of the $k$-mismatch longest common substrings is {\em bimodal},
with a first peak in the distribution corresponding to the background
matches and the second peak corresponding to the homologous matches. 
We show that the number of substitutions per positions can be easily
estimated from the position of this second peak. 

To simplify the following calculations, we 
ignore the longest exact match in equation (\ref{sum_kmacs}),
and consider only the length of the gap-free `extension' of this match.  
To model the length of these $k$-mismatch {\em extensions}, 
we define define random variables  
\begin{equation}
\label{kmacs_extension}
	\hat{X}^{(k)}_i = \tilde{X}_{i}^{(k+1)}
 - X_i =  X^{(k)}_{i+X_i+1,j^*+X_i+1}    
\end{equation}
In other words, for a position $i$ in sequence $S_1$, we are looking for
the longest substring starting at $i$ that exactly matches a substring
of $S_2$. If $j^*$ is the starting position of this substring of $S_2$, 
we define $\hat{X}^{(k)}_i$ as the length of the longest possible
substring of $S_1$ starting at position $i+ X_i + 1$ 
that matches a substring of $S_2$ 
starting at position  $j^* + X_i + 1$ with a Hamming distance of up to $k$.

\begin{theorem}
\label{prob_increase}
Let $\hat{X}^{(k)}_i$ be defined as in (\ref{kmacs_extension}). Then $\hat{X}^{(k)}_i$
is the sum of two unimodal distributions, the a `homologous' and a `background'
contribution, and the maximum of the `homologous' contribution  is reached at
\[ m_H = \left\lceil \frac{k}{1-p} -1 \right\rceil \]
and the maximum of the `background contribution' is reached at 
\[ m_B = \left\lceil \frac{k}{1-q} -1 \right\rceil \]

\end{theorem} 

\begin{proof} 
As in (\ref{binom_eq}), the distribution of $\hat{X}^{(k)}_i$ 
conditional on $j^*=i$ or $j^*\not=i$, respectively, can be easily calculated 
as 
\begin{equation}
\nonumber
	P\left( \hat{X}^{(k)}_i = m \middle| j^*=i \right) = 
	P\left( X^{(k)}_{i+ X_i+1,i+ X_i+1}  =  m \right)   
 =   {m \choose k} p^{m-k} (1-p)^{k+1} 
\end{equation}
and
\begin{equation}
\nonumber
	P\left( \hat{X}^{(k)}_i = m \middle| j^*\not= i \right)  
 =   {m \choose k} q^{m-k} (1-q)^{k+1} 
\end{equation}
so we have 
\begin{equation}
	\label{p_kmacs_ext} 
	\begin{aligned}
		P\left( \hat{X}^{(k)}_i = m\right)  
		&=  P(j^* = i)  {m \choose k} p^{m-k} (1-p)^{k+1} \\
		&+ P(j^*\not= i)  {m \choose k} q^{m-k} (1-q)^{k+1} 
	\end{aligned}
\end{equation}

For the {\em homologous} part 
\[ H_k(m) = {m \choose k} p^{m-k} (1-p)^{k+1}  \]
we obtain the recursion
\[ H_k(m+1)= P(j^*= i)  \frac{ (m+1)}{m+1-k}\cdot p \cdot H_k(m) \] 
so we have $H_k(m) < H_k(m+1)$ if and only if 
\begin{equation}
\label{increas_cond}
\frac{ m+1-k }{m+1} < p 
\end{equation}
Similarly, the `background contribution' 
\[ B_k(m) = P(j^*\not= i) {m \choose k} q^{m-k} (1-q)^{k+1}  \]
is increasing until
\begin{equation}
\nonumber
\frac{ m+1-k }{m+1} < q 
\end{equation}
holds, which concludes the proof of the theorem 
\end{proof} 

Theorem~\ref{prob_increase} gives us an easy way to estimate the match 
probability $p$:  
By inserting the second local maximum $m_{\max}$ of the 
empirical distribution of $\hat{X}_i$ into  
(\ref{increas_cond}),  we obtain
\begin{equation}
\label{p_estimate}
	\hat{p} \approx \frac{ m_{\max}+1-k }{m_{\max}+1} 
\end{equation}

\begin{figure}
\begin{center}
\includegraphics[width=\figwidths]{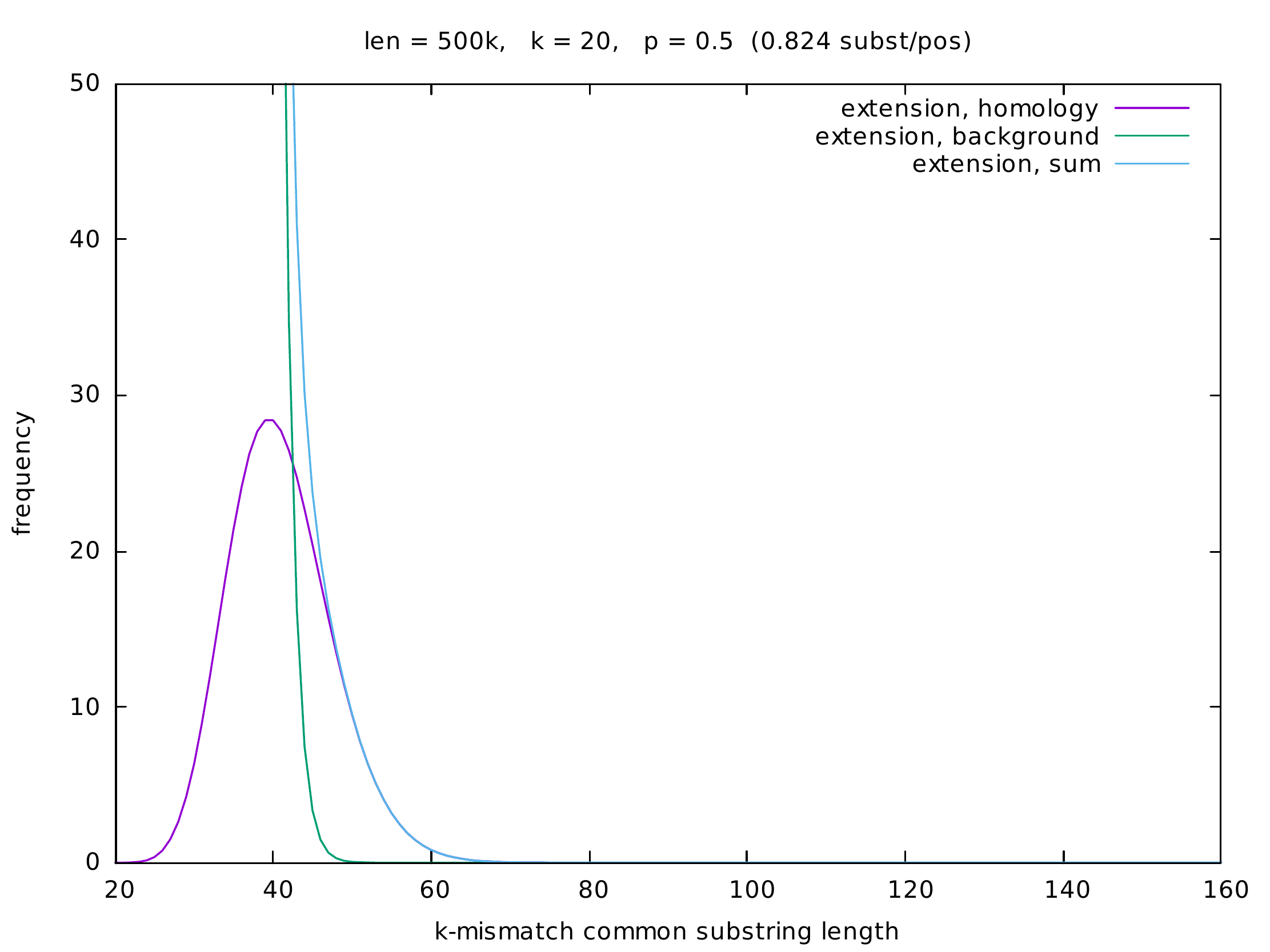}

\includegraphics[width=\figwidths]{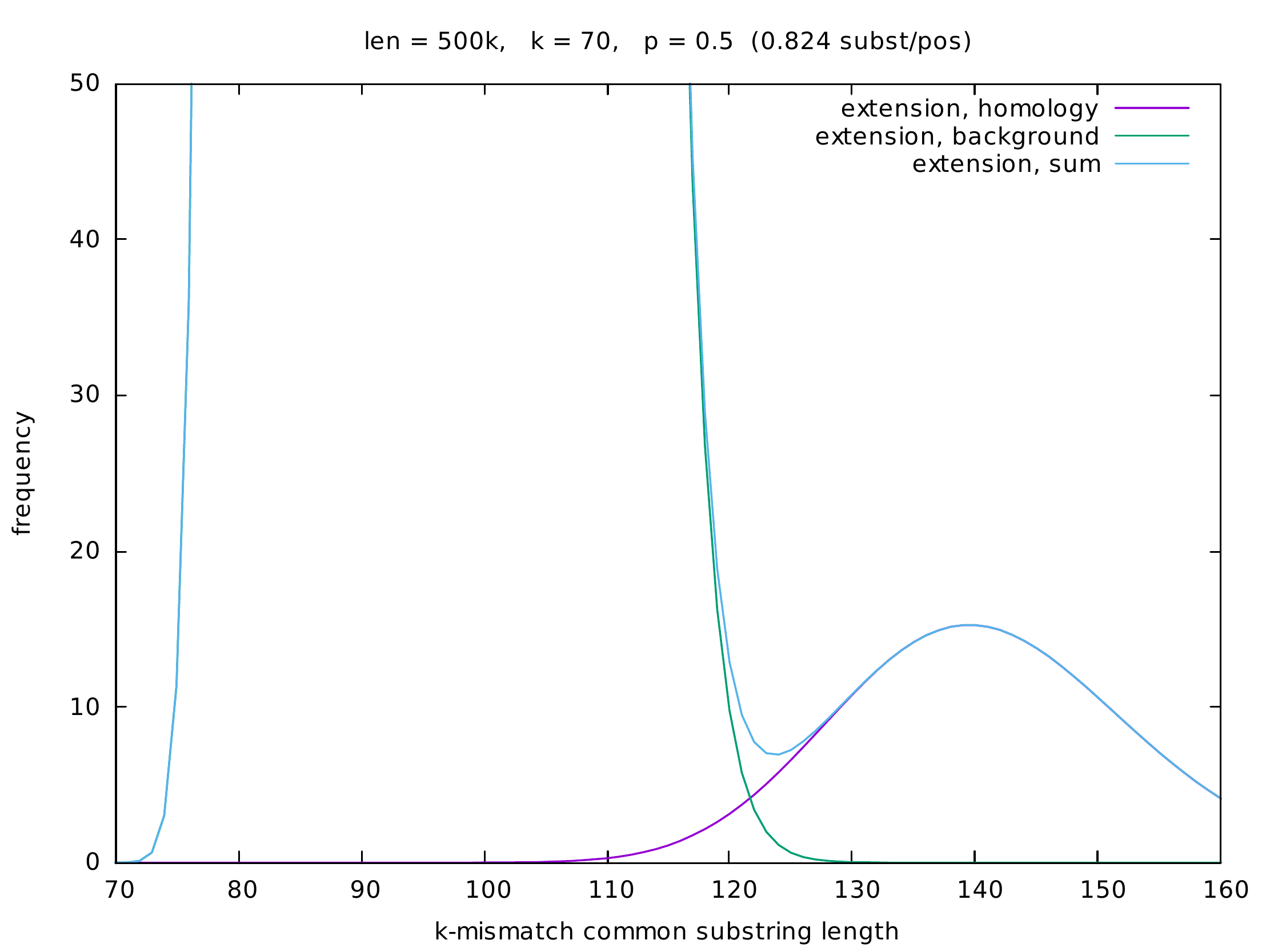}
\end{center}
\caption{\label{kvar}
Detail of the expected 
length distributions of the $k$-mismatch extensions in {\em kmacs}
for a pair of sequences of length $L=500$ {\em kb} with a match probability
of $p=0.5$ for $k=10$ (top) and $k=70$ (bottom). 
Expected frequencies were calculated using
equation~(\ref{p_kmacs_ext}), distinguishing between `homologous' and 
`background' matches.
A large enough value of~$k$ is necessary to detect the second peak in the
distribution that corresponds to the `homologous' matches.  
}
\end{figure}

\begin{figure}
\begin{center}
\includegraphics[width=\figwidth]{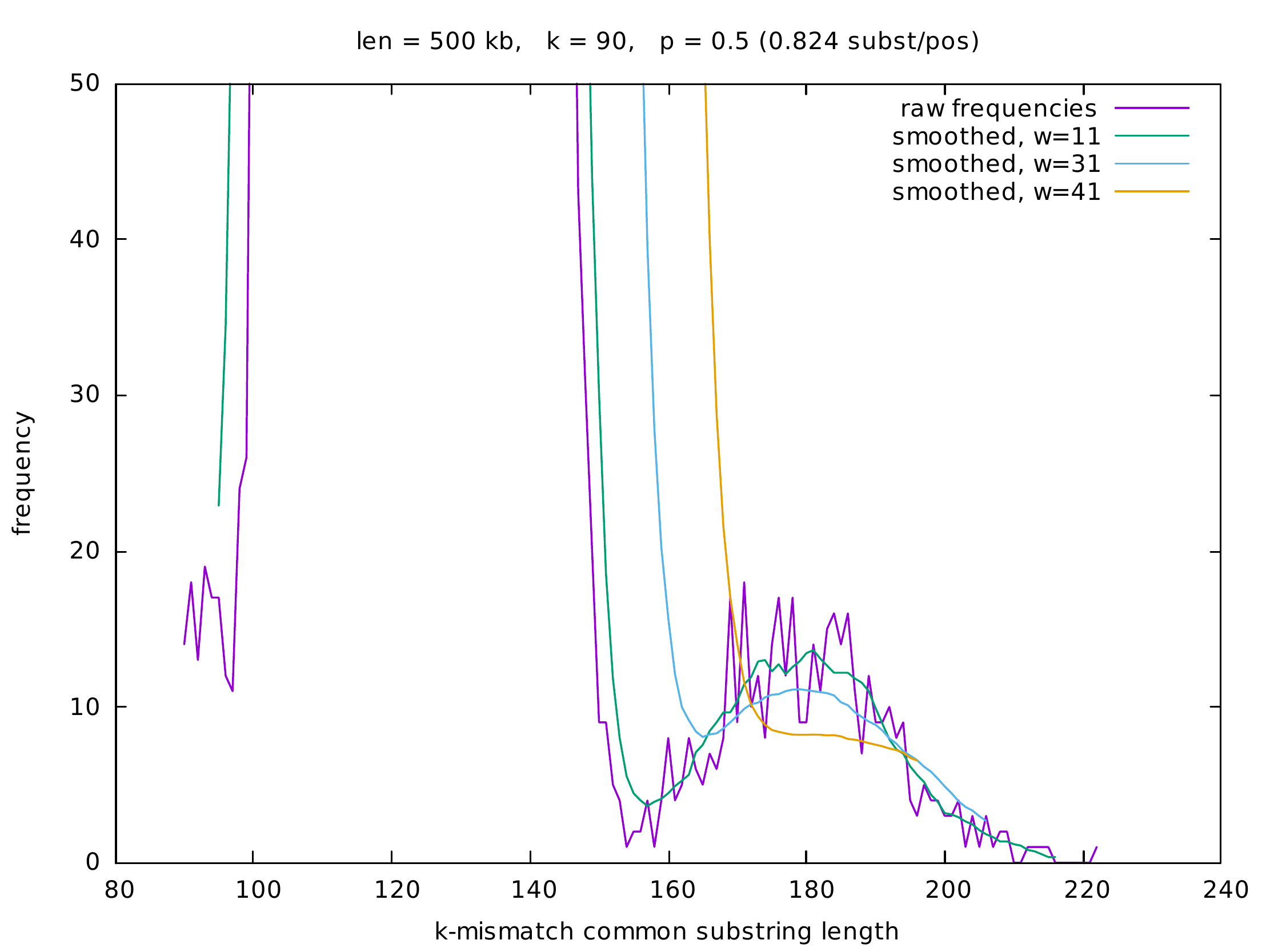}
\end{center}
\caption{\label{k90plot}
Detail of the length distribution of the $k$-mismatch extensions in {\em kmacs} 
for a pair of simulated DNA sequences of length $L=500$ {\em kb} with $k=90$. 
(raw frequencies and smoothed distributions). Different parameters were used for
for the width~$w$ of the smoothing window.  
The hight of the `homologous' peak is > 50,000}
\end{figure}

For completeness, we calculate 
the probability $P(j^* = i)$. First, we note that, for 
all $i$,  we have 
\begin{equation}
\nonumber
P\left(X_{i,j} < X_{i,i} \text{ for all } j\not=i\right) \le 
P(j^* = i) \le
P(X_{i,j} \le X_{i,i} \text{ for all } j\not=i)  
\end{equation}
and for all $m$ and $i\not= j$,  
\[ P(X_{i,j} < m) = 1-q^m\]
and
\[ P(X_{i,j} \le  m) = 1-q^{m+1}\]
hold. Thus, we obtain
\begin{equation}
\label{hom_prob}
\begin{aligned}
\MoveEqLeft{ P\left(X_{i,j} < X_{i,i} \text{ for all } j\not=i\right)  = \sum_m  P\left(X_{i,j} < X_{i,i} \text{ for all } j\not=i| X_{i,i} = m \right) P( X_{i,i} = m) }\\
 =& \sum_m  P\left(X_{i,j} < m  \text{ for all } j\not=i \right) P( X_{i,i} = m) \\  
 =& \sum_m \prod_{j\not=i}  P( X_{i,j} < m)    P( X_{i,i} = m)  \\
        &=  \sum_m  (1-q^m)^{L-1}  p^m (1-p) 
\end{aligned}
\end{equation}
and similarly
\begin{equation}
 P\left(X_{i,j} \le X_{i,i} \text{ for all } j\not=i\right)  =  \sum_m  (1-q^{m+1})^{L-1}  p^m (1-p) 
\end{equation}

\section{Implementation}
\label{implementation}
For each position $i$ in one of two input sequences, {\em kmacs} first 
calculates the length of the longest substring starting at $i$ that
exactly matches a substring of the other sequence. 
For a user-defined parameter $k$, the program then calculates the
length of the longest possible gap-free extension with up to $k$ mismatches
of this exact hit.  
The original version of the program uses the average length of these
$k$-mismatch common substrings (the initial exact match plus the $k-1$-mismatch
extension after the first mismatch) 
to calculate a distance between two sequences. 
We modified {\em kmacs} to output the length of the {\em extensions}
of the identified exact matches. Thus,  
to find $k$-mismatch common substrings, we ran {\em kmacs} with parameter
$k+1$, and we consider the length of the $k$-mismatch extension  
{\em after} the first mismatch.  
For each possible length $m$, the modified program outputs
the number $N(m)$ of $k$-mismatch extensions of length $m$, starting
after the first mismatch after the respective longest exact match.

\begin{figure}
\begin{center}
\includegraphics[width=11cm]{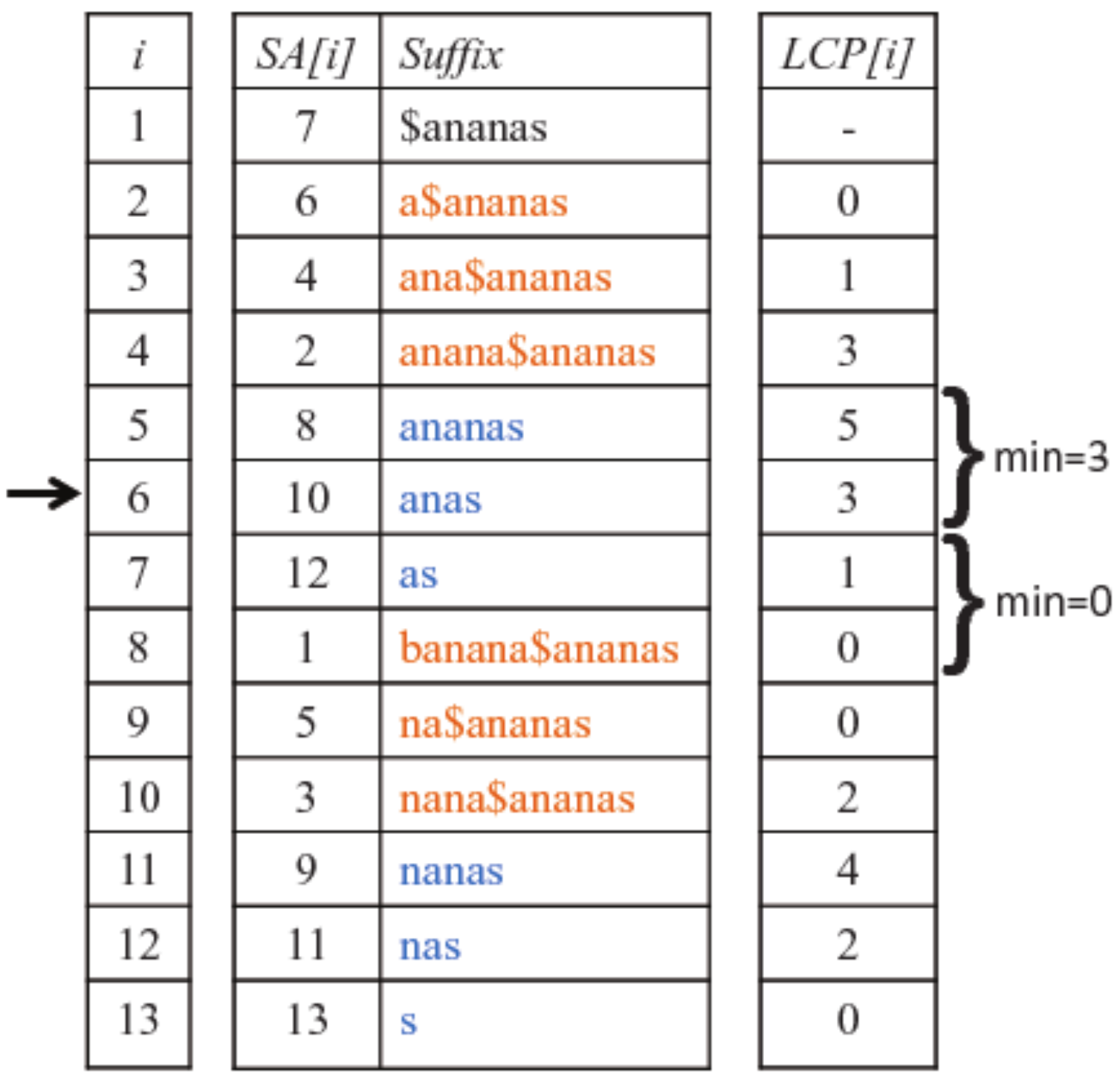}
\end{center}
\caption{\label{esa}Enhanced suffix array for sequences {\tt banana}
and {\tt ananas}. Suffixes of the concatenated sequence are lexicographically
ordered; a {\em longest common prefix (LCP)} array indicates the length of the
longest common prefix of a suffix with its predecessor in the list
(Figure from~\cite{lei:mor:14}).}  
\end{figure}

To find for each position $i$ in one sequence the length of the 
longest string at $i$ 
matching a substring of the other sequences, {\em kmacs} uses a standard
procedure based on  
{\em enhanced suffix arrays} \cite{man:mye:90}, see Figure~\ref{esa}.
To find the longest exact match starting at $i$, the algorithm goes to the
corresponding position in the suffix array. It then goes in both directions,
up and down, in the suffix array until the first entry from the respective
other sequence is found. In both cases, the minimum of the {\em LCP} values
is recorded. The maximum of these two minima is the length of the longest 
substring in the other sequence matching a substring starting at $i$. 
In Figure~\ref{esa}, for example, if $i$ is position 3 in the string
{\tt ananas}, i.e. the 10th position in the concatenate string, 
the minimum {\em LCP} value 
until the first entry from {\tt banana} 
is found, is 3 if one goes up the array and 0 if one goes down. Thus,
the longest string in {\tt banana} matching a substring starting at 
position~3 in {\tt ananas} has length~3.  

Note that, for a position $i$ in one sequence, it is possible that there 
exist more than one maximal 
substring in the other sequence matching a substring at $i$. In this case,
our modified algorithm uses {\em all} of these maximal substring
matches, {\em i.e.} all maximal exact string matches are extended as 
described above. All these hits can be easily found in the suffix array by
extending the search in upwards or downwards direction until the minimum
of the {\em LCP} entries decreases. In the above  
example, there is a second occurrence of {\tt ana}
in {\tt banana} which is found by moving one more position upwards (the
corresponding {\em LCP} value is still 3).

In addition, we modified the original {\em kmacs} to ensure that for each pair
$(i,j)$ of positions from the two input sequences, only {\em one}  single
extended $k$-mismatch common substring is considered. The rationale behind this
is as follows: if the two input sequences share a long common substring $S$, 
then there will be many positions $i$ in the first sequence within $S$ 
such that the longest exact string match at $i$ matches to
a substring in $S$ in the second sequence. Thus, all these exact substring
matches are identical up to different starting positions, so they end at 
the same first mismatch between  $S_1$ and $S_2$. Consequently, 
the $k$-mismatch extensions
of these exact matches are all exactly the same. As a result, for real-world
sequences with long exact substrings, isolated positions $m$ in the
length distribution of the $k$-mismatch common substrings can be observed
with very large values $N(m)$ while $N(m')=0$ for other values $m'$ around 
$m$.

To further process the length distribution returned by the modified {\em kmacs},
we implemented a number of {\em Perl} scripts. 
First, the length distribution of the $k$-mismatch common substrings is
smoothed using a window of length~$w$. 
Next, we search for the second local maximum in this smoothed length
distribution. This second peak should represent the {\em homologous} 
$k$-mismatch common substrings, while the first, larger peak represents
the {\em background} matches, see Figures~\ref{k90plot} and 
\ref{kvar}.
A simple script identifies the position $m^*$ of the second highest local peak 
under two side constraints: we require the height $N(m^*)$ of the second peak to
be substantially smaller than the global maximum, and we required for
that $N(m^*)$ is larger than $N(m^*-x)$. Quite arbitrarily, we required the
second peak to be 10 times smaller than the global maximum peak, and we
used a value of $x=4$. These constraints were introduced to prevent the 
program to identify small side peaks within the background peak. 

Finally, we use the position $m^*$ of the second largest peak in the
smoothed length distribution of $k$-mismatch common substrings 
to estimate the match probability~$p$ in an alignment of the two input
sequences using expression (\ref{p_estimate}). 
The usual {\em Jukes-Cantor} correction is then used to 
estimate the number of substitutions per position that have occurred
since the two sequences separated from their last  common ancestor.

We should mention that our algorithm is not always able to output a distance
value for two input sequences. 
It is possible that the algorithm fails to find a second maximum in the length
distribution of the $k$-mismatch common substrings, so in these cases no
distance can be calculated.

\section{Test Results}
\label{results}
To evaluate our approach, we used simulated and real-world genome sequences.
As a first set of test data, we generated pairs of simulated DNA sequences
of length 500 {\em kb} with varying evolutionary distances
and compared the distances estimated with our algorithm -- {\em i.e.} the
estimated number of substitutions per position -- 
to their `real' distances. For each distance value, we generated 100 pairs
of sequences and calculated the average and standard deviation of the
estimated distance values.  
Figure~\ref{sim_results} shows the results of these test runs.  
 with a parameter 
$k=90$ and a  smoothing window size of $w=31$, with error bars representing 
standard deviations.
A program run on a pair of sequences of length 500 {\em kb} took less than
a second.

Figure~\ref{k90plot} shows a detail of the length distribution for one
of these sequence pairs with various values for~$w$. 
In Figure~\ref{sim_results}, 
the results are reported for a given distance value, if distances could be
computed for at least 75 out of the 100 sequence pairs. 
As can be seen in the figure, our approach
accurately estimates evolutionary distances  up to 0.9 around substitutions
per position. For larger distances, the program did not return a sufficient
number of distance values, so no results are reported here.  
To demonstrate the influence of the parameter~$k$, 
we plotted in Figure~\ref{kvar}, for a given set of parameters, 
the expected number of $k$-mismatch 
common substring extensions of length~$m$,  
calculated with equation~(\ref{p_kmacs_ext}), against~$m$. 

As a real-word test case, we used a set of 27 mitochondrial genomes from primates that has
been used as benchmark data in previous studies on alignment-free
sequence comparison. We applied our method with different values of $k$
and with different  window lengths $w$ for the smoothing. 
In addition, we ran the programs {\em andi} \cite{hau:klo:pfa:14}
and our previously published program {\em Filtered Spaced-Word
Matches (FSWM)}~\cite{lei:soh:mor:17} to these data. As a reference tree,
we used a tree calculated with {\em Clustal $\Omega$}~\cite{sie:wil:din:etal:11}
and {\em Neighbour Joining} \cite{sai:nei:87}. 
To compare the produced trees with this reference trees, we used
the {\em Robinson-Foulds} distance \cite{rob:fou:81} 
and the {\em branch score} distance \cite{kuh:fel:94}
as implemented in the {\em PHYLIP} program package \cite{fel:89}. 
Figure~\ref{mito_res} shows the performance of our approach with different 
parameter values and compares them to the results of {\em andi} and {\em FSWM}.
For the parameter values shown in the figure, our program was able to calculate
distances for all ${27 \choose2}=351$ pairs of sequences.  
The total run time to calculate the 351 distance values 
for the 27 mitochondrial genomes was less than 6 seconds.

\begin{figure}
\begin{center}
\includegraphics[width=\figwidth]{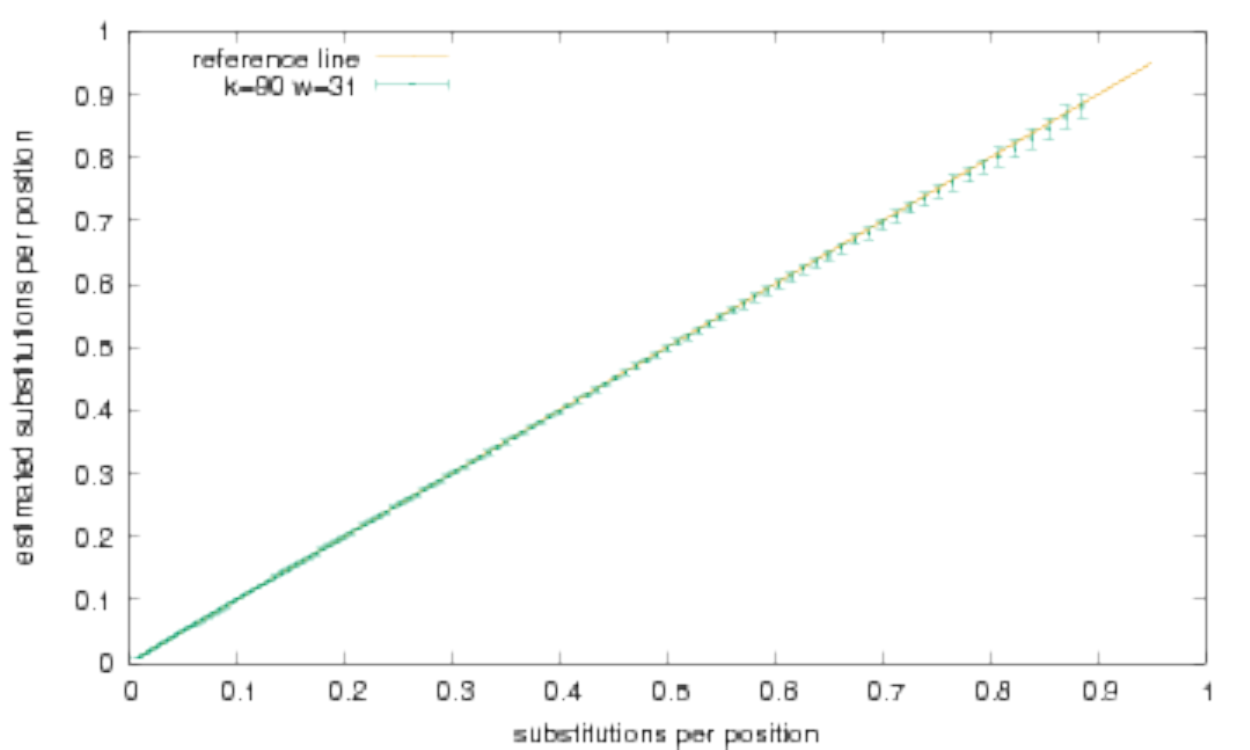}
\end{center}
\caption{\label{sim_results}Estimated distances -- {\em i.e.} estimated 
average number of substitutions per position -- 
for simulated sequence pairs, plotted against the `real' distances.
We used pairs of sequences of length $L=500$~{\em kb} and parameters $k=90$
and $w=31$.}
\end{figure}

\begin{center}
\begin{table}
\begin{tabular}{r|ccccccc} 
     &  k=30 &  k=50 &  k=70 &  k=90  &  k=120  &  k=150  & k=200 \\
\hline
 w=1 & 0.665 & 0.809 & 0.935 &  0.897 &  0.794  &  0.781  & 0.995 \\ 
 w=5 &    -  & 0.839 & 0.835 &  0.784 &  0.783  &  0.773  & 0.880 \\
w=11 &    -  &   -   & 0.869 &  0.808 &  0.788  &  0.781  & 0.863 \\
w=21 &    -  &   -   & 0.813 &  0.824 &  0.824  &  0.804  & 0.817 \\ 
w=31 &    -  &   -   & 0.813 &  0.824 &  0.824  &  0.829  & 0.835 \\
w=51 &    -  &   -   &    -  &      - &  0.824  &  0.819  & 0.820 \\

\end{tabular}
\caption{\label{table}
Distance values calculated with our algorithm for a pair of
simulated sequences of length $L=500$ {\em kb} with a match rate
of $p=0.5$,  corresponding to a distance of $0.824$ substitutions 
per position. Dashes indicate that no distance value could be calculated
since our algorithm could not find the second local maximum in the length
distribution of the $k$-mismatch common substrings.}
\end{table}
\end{center}

\begin{figure}
\begin{center}
\includegraphics[width=\figwidths]{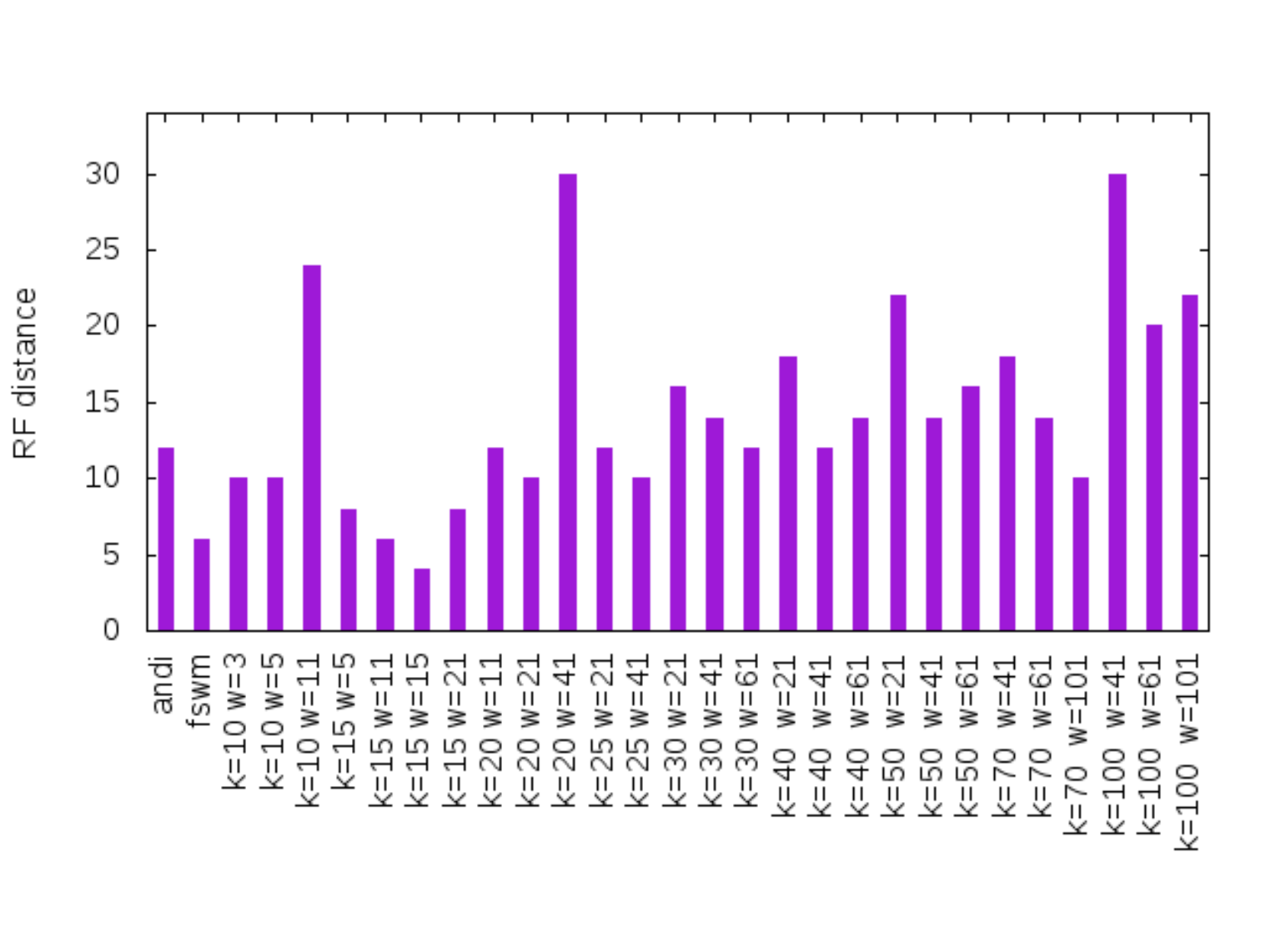}

\includegraphics[width=\figwidths]{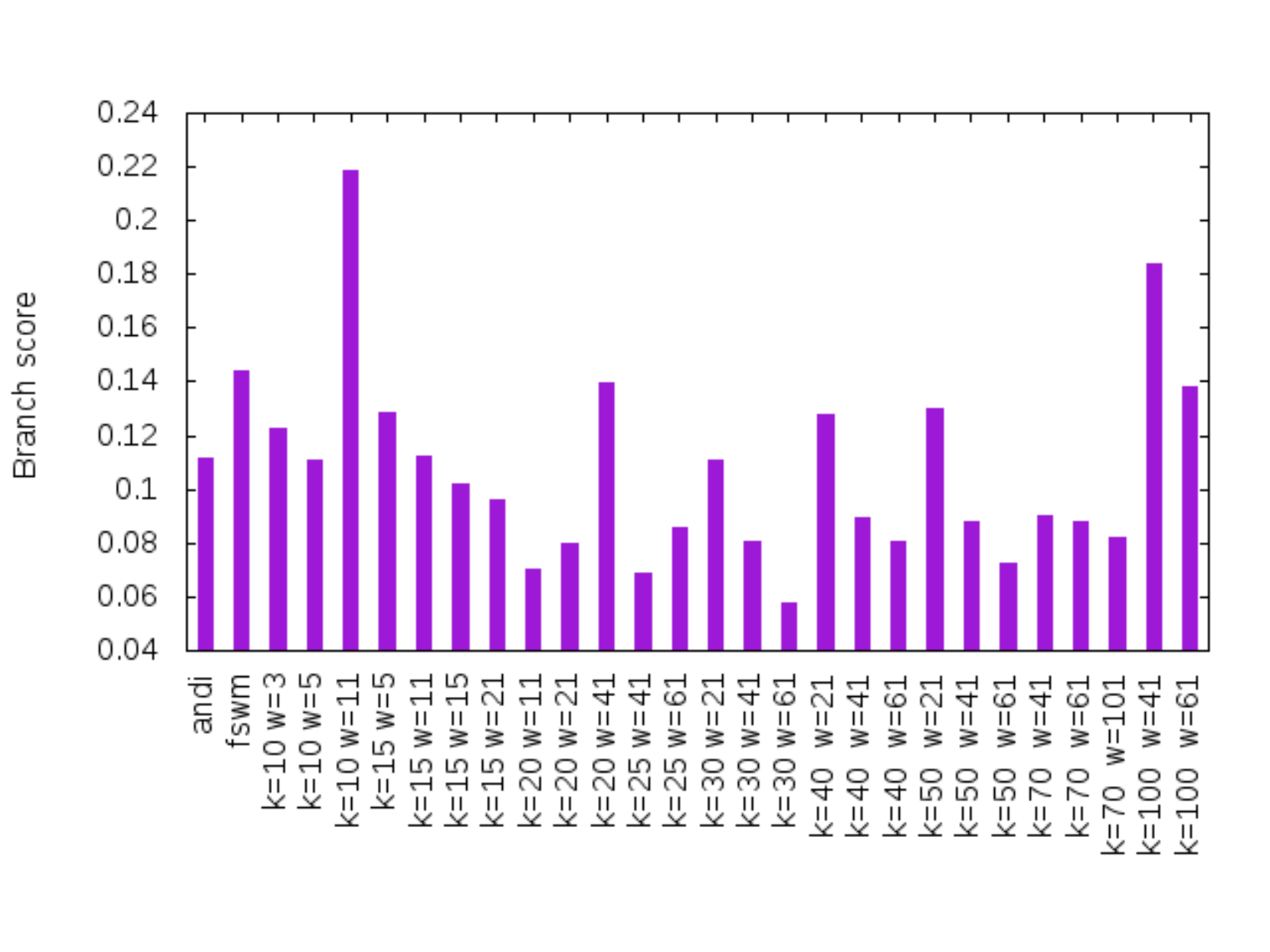}

\caption{\label{mito_res}Evaluation of various alignment-free 
methods for phylogeny reconstruction on on a set of 27 primate 
mitochondrial genomes. {\em Robinson-Foulds} distances (top)
and {\em branch scores} (bottom) 
were calculated to measure the difference 
between the resulting trees and a reference tree obtained with
 {\em Clustal}~$\Omega$ and {\em Neighbour Joining}.} 
\end{center}
\end{figure}

\section{Discussion}
\label{discussion}
In this paper, we introduced a new way of estimating phylogenetic
distances between genomic sequences. We showed that the average number
of substitutions per position since two sequences have separated from their
last common ancestor can be accurately estimated from the position of 
local maximum in the smoothed length
distribution of $k$-mismatch common substrings. 
To find this local maximum, we used a naive search procedure  on the smoothed
length distribution. 
Two parameter values have to be specified in our approach,  
the number $k$ of mismatches and the size $w$ of the smoothing window
for the length distribution.
Table~\ref{table} shows that our distance estimates are reasonably stable
for a range of values of $k$ and $w$.  

A suitable value of the parameter $k$ is important  to separate 
the `homologous' peak from the `background' peak in the length distribution
of the $k$-mismatch common substrings. 
As follows from theorem \ref{prob_increase}, the distance between these
two peaks is proportional to $k$. The value of~$k$ 
must be large enough to ensure that the homologous peak
has a sufficient distance to the background peak to be detectable, see
Figure~\ref{kvar}. 
Our data show, on the other hand, that our distance estimates become
less precise if $k$ is too large. 

Specifying a suitable size~$w$  of the smoothing window is also important
to obtain accurate distance estimates;  
a large enough window is necessary to 
avoid ending up in a local maximum of the raw length distribution.  
For the data shown in Figure~\ref{k90plot}, for example, 
our approach finds the second maximum 
of the length distribution at 179 if a window width of  $w=31$ is chosen.
From this value, the  match probability~$p$ is estimated as 
\[ \hat{p} = \frac{179+1-90}{179+1} = 0.5 \]
using equation (\ref{increas_cond}),  
corresponding to 0.824 substitutions per position according to 
the {\em Jukes-Cantor} formula. This was exactly the value that we 
used to generate this pair of sequences.   

With window lengths of $w=21$ and $w=1$ (no smoothing at all), however,
the second local maxima of the length distribution would be found at
181 and 171, respectively, leading to distance estimates of 0.808 ($w=11$) 
and 0.897 ($w=1$). 
If the width $w$ of the smoothing window is too large, on the other hand,  
the second peak may be obscured by the first `background' peak. In this case, no
peak is found and no distance can be calculated. In 
Figure~\ref{k90plot}, for example, this happens with if a window width $w=51$
is used. 
Further studies are necessary
to find out suitable values for $w$ and $k$, depending on the length
of the input sequences.

Finally, we should say that we used a rather naive way to identify possible
homologies that are then extended to find $k$-mismatch common substrings. 
As becomes obvious from the size of the homologous and background peaks
in our plots, our approach finds far more background matches than
homologous matches. Reducing the noise of background matches should 
help to find the position of the homologous peak in the length distributions. 
We will therefore explore alternative ways to find possible homologies
that can be used as starting points for $k$-mismatch common substrings.

\bibliographystyle{abbrv}



\begin{thebibliography}{10}

\bibitem{alu:apo:tha:15}
S.~Aluru, A.~Apostolico, and S.~V. Thankachan.
\newblock Efficient alignment free sequence comparison with bounded mismatches.
\newblock RECOMB'12, pages 1--12, 2015.

\bibitem{apo:gue:lan:piz:16}
A.~Apostolico, C.~Guerra, G.~M. Landau, and C.~Pizzi.
\newblock Sequence similarity measures based on bounded hamming distance.
\newblock {\em Theoretical Computer Science}, 638:76--90, 2016.

\bibitem{cha:law:94}
W.~I. Chang and E.~L. Lawler.
\newblock Sublinear approximate string matching and biological applications.
\newblock {\em Algorithmica}, 12:327--344, 1994.

\bibitem{cho:hor:lev:09}
B.~Chor, D.~Horn, Y.~Levy, N.~Goldman, and T.~Massingham.
\newblock Genomic {DNA} $k$-mer spectra: models and modalities.
\newblock {\em Genome Biology}, 10:R108, 2009.

\bibitem{com:ver:12}
M.~Comin and D.~Verzotto.
\newblock Alignment-free phylogeny of whole genomes using underlying subwords.
\newblock {\em Algorithms for Molecular Biology}, 7:34, 2012.

\bibitem{fel:89}
J.~Felsenstein.
\newblock {PHYLIP - Phylogeny Inference Package (Version 3.2)}.
\newblock {\em Cladistics}, 5:164--166, 1989.

\bibitem{gus:97}
D.~Gusfield.
\newblock {\em Algorithms on Strings, Trees, and Sequences: Computer Science
  and Computational Biology}.
\newblock Cambridge University Press, Cambridge, UK, 1997.

\bibitem{hah:lei:oun:etal:16}
L.~Hahn, C.-A. Leimeister, R.~Ounit, S.~Lonardi, and B.~Morgenstern.
\newblock {{\em rasbhari}: optimizing spaced seeds for database searching, read
  mapping and alignment-free sequence comparison}.
\newblock {\em PLOS Computational Biology}, 12(10):e1005107, 2016.

\bibitem{hau:klo:pfa:14}
B.~Haubold, F.~Kl{\"o}tzl, and P.~Pfaffelhuber.
\newblock andi: Fast and accurate estimation of evolutionary distances between
  closely related genomes.
\newblock {\em Bioinformatics}, 31:1169--1175, 2015.

\bibitem{hau:pfa:dom:wie:09}
B.~Haubold, P.~Pfaffelhuber, M.~Domazet-Loso, and T.~Wiehe.
\newblock Estimating mutation distances from unaligned genomes.
\newblock {\em Journal of Computational Biology}, 16:1487--1500, 2009.

\bibitem{hau:pie:moe:wie:05}
B.~Haubold, N.~Pierstorff, F.~M{\"o}ller, and T.~Wiehe.
\newblock Genome comparison without alignment using shortest unique substrings.
\newblock {\em BMC Bioinformatics}, 6:123, 2005.

\bibitem{hau:wie:06}
B.~Haubold and T.~Wiehe.
\newblock How repetitive are genomes?
\newblock {\em BMC Bioinformatics}, 7:541, 2006.

\bibitem{hoe:rig:rag:06}
M.~H{\"o}hl, I.~Rigoutsos, and M.~A. Ragan.
\newblock Pattern-based phylogenetic distance estimation and tree
  reconstruction.
\newblock {\em Evolutionary Bioinformatics Online}, 2:359--375, 2006.

\bibitem{juk:can:69}
T.~H. Jukes and C.~R. Cantor.
\newblock {\em Evolution of Protein Molecules}.
\newblock Academy Press, New York, 1969.

\bibitem{kuh:fel:94}
M.~K. Kuhner and J.~Felsenstein.
\newblock A simulation comparison of phylogeny algorithms under equal and
  unequal evolutionary rates.
\newblock {\em Molecular Biology and Evolution}, 11:459--468, 1994.

\bibitem{lei:bod:hor:lin:mor:14}
C.-A. Leimeister, M.~Boden, S.~Horwege, S.~Lindner, and B.~Morgenstern.
\newblock Fast alignment-free sequence comparison using spaced-word
  frequencies.
\newblock {\em Bioinformatics}, 30:1991--1999, 2014.

\bibitem{lei:den:mor:17}
C.-A. Leimeister, T.~Dencker, and B.~Morgenstern.
\newblock Anchor points for genome alignment based on filtered spaced word
  matches.
\newblock {\em arXiv:arXiv:1703.08792[q-bio.GN]}, 2017.

\bibitem{lei:mor:14}
C.-A. Leimeister and B.~Morgenstern.
\newblock {\em kmacs}: the $k$-mismatch average common substring approach to
  alignment-free sequence comparison.
\newblock {\em Bioinformatics}, 30:2000--2008, 2014.

\bibitem{lei:soh:mor:17}
C.-A. Leimeister, S.~Sohrabi-Jahromi, and B.~Morgenstern.
\newblock Fast and accurate phylogeny reconstruction using filtered spaced-word
  matches.
\newblock {\em Bioinformatics}, 33:971--979, 2017.

\bibitem{man:mye:90}
U.~Manber and G.~Myers.
\newblock Suffix arrays: a new method for on-line string searches.
\newblock {\em Proceedings of the first annual ACM-SIAM symposium on Discrete
  algorithms}, SODA '90:319--327, 1990.

\bibitem{mor:zhu:hor:lei:15}
B.~Morgenstern, B.~Zhu, S.~Horwege, and C.-A. Leimeister.
\newblock Estimating evolutionary distances between genomic sequences from
  spaced-word matches.
\newblock {\em Algorithms for Molecular Biology}, 10:5, 2015.

\bibitem{noe:17}
L.~No{\'e}.
\newblock Best hits of 11110110111: model-free selection and parameter-free
  sensitivity calculation of spaced seeds.
\newblock {\em Algorithms for Molecular Biology}, 12:1, 2017.

\bibitem{pet:gue:piz:17}
U.~F. Petrillo, C.~Guerra, and C.~Pizzi.
\newblock A new distributed alignment-free approach to compare whole proteomes.
\newblock {\em Theoretical Computer Science}, in press, 2017.

\bibitem{piz:16}
C.~Pizzi.
\newblock {MissMax}: alignment-free sequence comparison with mismatches through
  filtering and heuristics.
\newblock {\em Algorithms for Molecular Biology}, 11:6, 2016.

\bibitem{rob:fou:81}
D.~Robinson and L.~Foulds.
\newblock Comparison of phylogenetic trees.
\newblock {\em Mathematical Biosciences}, 53:131--147, 1981.

\bibitem{sai:nei:87}
N.~Saitou and M.~Nei.
\newblock The neighbor-joining method: a new method for reconstructing
  phylogenetic trees.
\newblock {\em Molecular Biology and Evolution}, 4:406--425, 1987.

\bibitem{sie:wil:din:etal:11}
F.~Sievers, A.~Wilm, D.~Dineen, T.~J. Gibson, K.~Karplus, W.~Li, R.~Lopez,
  H.~McWilliam, M.~Remmert, J.~S{\"o}ding, J.~D. Thompson, and D.~G. Higgins.
\newblock Fast, scalable generation of high-quality protein multiple sequence
  alignments using {Clustal Omega}.
\newblock {\em Molecular Systems Biology}, 7:539, 2011.

\bibitem{sim:jun:wu:kim:09}
G.~E. Sims, S.-R. Jun, G.~A. Wu, and S.-H. Kim.
\newblock Alignment-free genome comparison with feature frequency profiles
  ({FFP}) and optimal resolutions.
\newblock {\em Proceedings of the National Academy of Sciences},
  106:2677--2682, 2009.

\bibitem{tha:apo:alu:16}
S.~V. Thankachan, A.~Apostolico, and S.~Aluru.
\newblock A provably efficient algorithm for the $k$-mismatch average common
  substring problem.
\newblock {\em Journal of Computational Biology}, 23:472--482, 2016.

\bibitem{tha:cho:liu:apo:alu:15}
S.~V. Thankachan, S.~P. Chockalingam, Y.~Liu, A.~Apostolico, and S.~Aluru.
\newblock {ALFRED}: a practical method for alignment-free distance computation.
\newblock {\em Journal of Computational Biology}, 23:452--460, 2016.

\bibitem{tha:cho:liu:etal:17}
S.~V. Thankachan, S.~P. Chockalingam, Y.~Liu, A.~Krishnan, and S.~Aluru.
\newblock A greedy alignment-free distance estimator for phylogenetic
  inference.
\newblock {\em BMC Bioinformatics}, 18:238, 2017.

\bibitem{uli:bur:tul:cho:06}
I.~Ulitsky, D.~Burstein, T.~Tuller, and B.~Chor.
\newblock The average common substring approach to phylogenomic reconstruction.
\newblock {\em Journal of Computational Biology}, 13:336--350, 2006.

\bibitem{vin:14}
S.~Vinga.
\newblock Editorial: Alignment-free methods in computational biology.
\newblock {\em Briefings in Bioinformatics}, 15:341--342, 2014.

\bibitem{vin:car:fra:etal:12}
S.~Vinga, A.~M. Carvalho, A.~P. Francisco, L.~M.~S. Russo, and J.~S. Almeida.
\newblock Pattern matching through {Chaos Game Representation}: bridging
  numerical and discrete data structures for biological sequence analysis.
\newblock {\em Algorithms for Molecular Biology}, 7:10, 2012.

\bibitem{yi:jin:13}
H.~Yi and L.~Jin.
\newblock Co-phylog: an assembly-free phylogenomic approach for closely related
  organisms.
\newblock {\em Nucleic Acids Research}, 41:e75, 2013.

\end{thebibliography}

\end{document}